\documentclass[orivec,oribibl]{llncs}
\usepackage{llncsdoc}

\usepackage{setspace}
\usepackage{amsmath,amssymb}
\usepackage{color}
\usepackage{fullpage}
\newtheorem{fact}{Fact}

\newcommand{\remove}[1]{}

\newcommand{\pr}[2][{}]{\mathop{\mathrm{Pr}}_{{#1}}\!\left[ #2 \right]}
\newcommand{\expect}[2][{}]{\mathop{\mathbb{E}}_{#1}\!\left[ #2 \right]}
\newcommand{\set}[1]{\left\{#1\right\}}

\newcommand{\ceil}[1]{\left\lceil #1 \right\rceil}

\newcommand{\RLS}{\mathsf{RLS}} 
\newcommand{\QLS}{\mathsf{QLS}} 

\newcommand{\newfix}[1]{\textcolor{black}{#1}}
\newcommand{\fix}[1]{\textcolor{black}{#1}}

\title{Quantum and Randomized Lower Bounds for\\ 
Local Search on Vertex-Transitive Graphs}

\author{{\sc Hang Dinh} \and {\sc Alexander Russell}}

\institute{Department of Computer Science \& Engineering\\
University of Connecticut\\
Storrs, CT 06269, USA\\
\email{\{hangdt, acr\}@engr.uconn.edu}
}

\date{\today}               

\begin{document}
\maketitle

\begin{abstract}
We study the problem of \emph{local search} on a graph. Given a real-valued black-box function $f$ on the graph's vertices, this is the problem of determining a local minimum of $f$---a vertex $v$ for which $f(v)$ is no more than $f$ evaluated at any of $v$'s neighbors. In 1983, Aldous gave the first strong lower bounds for the problem, showing that any randomized algorithm requires  $\Omega(2^{n/2 - o(1)} )$ queries to determine a local minima on the $n$-dimensional hypercube. The next major step forward was not until 2004 when Aaronson, introducing a new method for query complexity bounds, both strengthened this lower bound to $\Omega(2^{n/2}/n^2)$ and gave an analogous lower bound on the quantum query complexity. While these bounds are very strong, they are known only for narrow families of graphs (hypercubes and grids). We show how to generalize Aaronson's techniques in order to give randomized (and quantum) lower bounds on the query complexity of local search for the family of vertex-transitive graphs. In particular, we show that for any vertex-transitive graph $G$ of $N$ vertices and diameter $d$, the randomized and quantum query complexities for local search on $G$ are $\Omega\left(\frac{\sqrt{N}}{d\log N}\right)$ and $\Omega\left(\frac{\sqrt[4]{N}}{\sqrt{d\log N}}\right)$, respectively.
\end{abstract}

\section{Introduction}
The \emph{local search} problem is that of determining a local minimum of a function defined on the vertices of a graph. Specifically, given a real-valued black-box function $f$ on the vertices of a graph, this is the problem of determining a vertex $v$ at which $f(v)$ is no more than $f$ evaluated at any of $v$'s neighbors. 
The problem provides an abstract framework for studying local search heuristics that have been
widely applied in combinatorial optimization, heuristics that typically combine random selection with steepest descent.
The performance of these heuristic algorithms, as recognized in \cite{Ref_Aardal97decade}, \emph{``was generally considered to be satisfactory, partly
based on experience, partly based on a belief in some physical or biological analogy, \ldots''}
Ideally, of course, we would evaluate practical results in the context of crisp theoretical bounds on the complexity of these problems! Moreover, as pointed out in~\cite{Ref_Aaronson04lower}, the complexity of the local search problem
is also central for understanding a series of complexity classes which are subclasses of the \emph{total function} class $\mathsf{TFNP}$, 
including $\mathsf{PPP}$ (Polynomial Pigeohole Principle),
$\mathsf{PODN}$ (Polynomial Odd-Degree Node), and $\mathsf{PLS}$ (Polynomial Local Search).

Local search has been the subject of a sizable body of theoretical work, in which complexity is typically measured by \emph{query complexity}: the total number of
queries made to the black-box function $f$ in order to find a local minimum. The first strong lower bounds were established in 1983 by Aldous~\cite{Ref_Aldous83minimization}, who showed that $2^{n/2-o(n)}$ queries are necessary, in general, in order for a randomized algorithm to find a local minimum of a function on the hypercube $\{0,1\}^n$. His proof constructs a rich collection of unimodal functions (that is, functions with a unique minimum) using hitting times of random walks. Llewellyn et al.~\cite{Ref_Llewellyn89local} improved the bound for deterministic query complexity to 
$\Omega(2^n/\sqrt{n})$ using an adversary argument characterized by vertex cuts.

With the advent of quantum computing, these black-box problems received renewed interest~\cite{Ref_Aaronson04lower,Ref_Zhang06new,Ref_Santha04quantum,Ref_Sun06quantum,Ref_Verhoeven06enhanced}. 
Most notably, Aaronson~\cite{Ref_Aaronson04lower} introduced a query lower bound method tuned for such problems, the \emph{relational adversary method}. Though his principal motivation was, no doubt, to provide \emph{quantum} lower bounds for local search, his techniques felicitously demonstrated improved bounds on randomized query complexity. In particular, he established a $\Omega(2^{n/2}n^{-2})$ lower bound for randomized local search
on the Boolean hypercube $\{0,1\}^n$ and the first nontrivial lower bound of $\Omega\left(n^{d/2-1}/ (d\log n)\right)$
for randomized local search on a $d$-dimensional grid $[n]^d$ with $d\geqslant 3$. 

These two lower bounds of Aaronson's have been recently improved by Zhang~\cite{Ref_Zhang06new}: refining Aaronson's framework, he established randomized query complexity lower bounds of $\Theta(2^{n/2}n^{1/2})$ on the hypercube and $\Theta(n^{d/2})$ on the grid $[n]^d$, $d\geqslant 4$. \fix{Additionally, Zhang's method can be applied to certain classes of product graphs, though it provides a rather complicated relationship between the lower bound and the product decomposition.\footnote{
Zhang \cite{Ref_Zhang06new}'s general lower bounds for a product graph $G_w\times G_c$ involve the length $L$ of the longest self-avoiding path in the ``clock'' graph $G_c$, and parameters $p(u,v,t)$'s of a regular random walk $W$ on $G_w$, where $p(u,v,t)$ is the probability that the random walk $W$ starting at $u$ ends up at $v$ after exactly $t$ steps.
In particular, he showed that $\RLS(G_w\times G_c)=\Omega\left(\frac{L}{\sum_{t=1}^{L/2} \max_{u,v}p(u,v,t)}\right)$ and $\QLS(G_w\times G_c)=\Omega\left(\frac{L}{\sum_{t=1}^{L/2} \sqrt{\max_{u,v}p(u,v,t)}}\right)$.}} A remaining hurdle in this direct line of research was to establish strong bounds for grids of small dimension.
Sun and Yao \cite{Ref_Sun06quantum} have \newfix{addressed} this problem, proving that the quantum query complexity is $\Omega \left( {n^{1/2 - c } } \right)$ for $[n]^2$ and $\Omega (n^{1- c } )$ for $[n]^3$, \newfix{for any fixed constant $c>0$.}
Focusing on general graphs, Santha and Szegedy~\cite{Ref_Santha04quantum} established quantum lower bounds of $\Omega(\log N)$ and 
$\Omega\left({\sqrt[8]{\tfrac{s}{\delta}}}/{\log N}\right)$ for local search on connected $N$-vertex graphs with maximal degree $\delta$ and \emph{separation number} $s$.\footnote{Santha and Szegedy define the \emph{separation number} $s(G)$ of a graph $G=(V,E)$ to be: 
$s(G)=\max\limits_{H\subset V} \min\limits_{S\subset H, |H|/4\leq |S|\leq 3|H|/4} |\partial_{H}S|$, where $\partial_{H}S=\set{v\in H\setminus S:\exists u\in S, (v,u)\in E}$ is the boundary of $S$ in the subgraph of $G$ restricted to $H$.} We remark that $s / \delta \leq \newfix{N}$ and these bounds for general graphs are, naturally, much weaker than those obtained for the highly structured families of graphs above.

In this article, we show off the flexibility of Aaronson's framework by extending it to arbitrary vertex-transitive graphs. Recall that a graph $G = (V, E)$ is \emph{vertex-transitive} if the automorphism group of $G$ acts transitively on the vertices: for any pair of vertices $x, y \in V$, there is a graph automorphism $\phi: V \rightarrow V$ for which $\phi(x) = y$. In particular, all Cayley graphs are vertex-transitive, so this class of graphs contains the \newfix{hypercubes of previous interest and the looped grids (tori)}.

Our lower bounds depend only on the size and diameter of the graph:

\begin{theorem}\label{Theorem_Vertex-Transitive}
Let $G$ be a connected, vertex-transitive graph with $N$ vertices and diameter $d$. Then 
\[
\RLS(G) = \Omega\left(\frac{\sqrt{N}}{d\log N}\right), \qquad\text{and}\qquad
\QLS(G) = \Omega\left(\frac{\sqrt[4]{N}}{\sqrt{d\log N}}\right)
\]
where $\RLS(G)$ and $\QLS(G)$ are the randomized and quantum query complexities of local search on $G$,  respectively.
\end{theorem}

Thus the vertex transitive graphs, compromising between the specific families of graphs addressed by~\cite{Ref_Aaronson04lower,Ref_Zhang06new} and the general results of Santha and Szegedy, still provide enough structure to support strong lower bounds.



\section{Definitions and Notation}
As in \cite{Ref_Aaronson04lower,Ref_Zhang06new}, we focus on the local search problem stated precisely as follows: given a graph $G=(V,E)$ and a black-box function $f: V\to \mathbb{R}$, find a local minimum of $f$ on $G$, i.e. find a vertex $v\in V$ such that $f(v) \leq f(w)$ for all neighbors $w$ of $v$. While the graph is known to the algorithm, the values of $f$ may only be \newfix{accessed} through an oracle. For an algorithm $\mathcal{A}$ that solves the local search problem on $G$, let $T(\mathcal{A}, G)$ be the maximum number of queries made to the black-box function by $\mathcal{A}$ before it returns a local minimum, this maximum taken over all functions $f$ on $G$. Given a graph $G$, the randomized query complexity for Local Search on $G$ is defined as $\min_{\mathcal{A}}T(\mathcal{A}, G)$, where the minimum ranges over all randomized algorithms $\mathcal{A}$ \fix{that output a local minimum with probability at least $2/3$}. The quantum query complexity is defined similarly, except that in the quantum case, $T(\mathcal{A}, G)$ is the maximum number of unitary query transformations \fix{of the error-bounded quantum algorithm $\mathcal{A}$}. The randomized (resp. quantum) query complexity for local search on $G$ will be denoted by  $\RLS(G)$ (resp. $\QLS(G)$).

As mentioned in the introduction, we focus on the vertex-transitive graphs, those whose automorphism groups act transitively on their vertex sets. Perhaps the most important subclass of the vertex-transitive graphs are the Cayley graphs. Let $G$ be a group (finite, in this article, and written multiplicatively) and $\Gamma$ a set of generators for $G$. The Cayley graph $C(G,\Gamma)$ is the graph with vertex set $G$ and edges $E = \{ (g, g \gamma) \mid \fix{g\in G},\gamma \in \Gamma \cup \Gamma^{-1}\}$. Note that with this definition for the edges, $(a, b) \in E \Leftrightarrow (b, a) \in E$ even when $G$ is nonabelian, and we may consider the graph to be undirected.

If $X$ is a sequence of vertices in a graph, we write $X_{i\to j}$ to denote the subsequence of $X$ from position $i$ to position $j$ ($i\leq j$).
If $X=(x_1,\ldots,x_t)$ is a sequence of vertices in a Cayley graph and $g$ is a group element, 
then we use $gX$ to denote the sequence $(gx_1,\ldots,gx_t)$. More generally,
for any automorphism $\sigma$ of a vertex-transitive graph $G$ and 
any sequence $X=(x_1,\ldots,x_t)$ of vertices in $G$, we let $\sigma X$ denote the sequence
$(\sigma(x_1),\ldots,\sigma(x_t))$.

The distance between two vertices $u,v$ of a graph $G$ shall be denoted by $\Delta_G(u,v)$;
when $G$ is understood from context we abbreviate to $\Delta(u,v)$.
The statistical distance between two distributions $D_1$ and $D_2$ on the same set $\Omega$ is defined as the distance in total variation:
$$
\|D_1-D_2\|_{\rm t.v.}=\max_{E\subset\Omega}|D_1(E)-D_2(E)| = \frac{1}{2} \sum_{\omega \in \Omega} | D_1(\omega) - D_2(\omega)| \,.
$$ 
We say that the distribution $D_1$ is \emph{$\delta$-close} to distribution $D_2$ if $\|D_1-D_2\|_{\rm t.v.}\leq \delta$.

\section{Generalizing Aaronson's Snakes}

Aaronson's~\cite{Ref_Aaronson04lower} application of the quantum and relational adversary methods to local search problems involved certain families of walks on a graph he called ``snakes.''
We begin by presenting Aaronson's snake method, adjusted to suit our generalization.
Throughout this section, let $G$ be a graph. A \emph{snake} $X$ of length $L$ is a sequence $(x_0,\ldots,x_L)$ of vertices in $G$
such that each $x_{i+1}$ is either equal to $x_i$ or a neighbor of $x_i$.
The subsequence $X_{0\to j}$ shall be referred to as the $j$-length ``head'' of the snake $X$.
Suppose $\mathcal{D}_{x_0,L}$ is a distribution over snakes of length $L$ starting at $x_0$, 
and $X$ is a snake drawn from $\mathcal{D}_{x_0,L}$.
In Aaronson's parlance, the snake $X$ ``flicks'' its tail by choosing a position $j$ uniformly at random from 
the set $\set{0,\ldots, L-1}$, and then drawing a new snake $Y$ from $\mathcal{D}_{x_0,L}$
conditioned on the event that $Y_{0\to j}=X_{0\to j}$, that is, that $Y$ has the same $j$-length head as $X$. In order to simplify the proof \fix{for vertex-transitive graphs} below, we consider a generalization in which a snake flicks its tail according to a  distribution $\mathbf{D}_L$, \fix{which may be nonuniform}, on the set $\set{0,\ldots, L-1}$. We shall relax, also, Aaronson's original condition that, aside from adjacent repetition of a vertex $v$, snakes be non-self-intersecting.

Let $X=(x_0,\ldots,x_L)$ be a snake. 
Define the function $f_X$ on $G$ as follows: for each vertex $v$ of $G$,
$$
f_X(v)=
\begin{cases}
	L-\max\set{i:x_i=v} & \text{if $v\in X$,} \\
       L+\Delta(x_0,v) & \text{if $v\not\in X$\,.}
\end{cases}
$$
In other words, $f_X(x_L)=0$, and for any $i<L$, $f_X(x_i)=L-i$ if $x_i\not\in\set{x_{i+1},\ldots,x_L}$.
Clearly $f_X$ has a unique local minimum at $x_L$.

Let $X$ and $Y$ be snakes of length $L$ starting at $x_0$. 
A vertex $v$ is called a \emph{disagreement} between $X$ and $Y$ if $v\in X\cap Y$ and $f_X(v)\neq f_Y(v)$.
We say $X$ and $Y$ are \emph{consistent} if there is no disagreement between $X$ and $Y$. 
Observe that so long as $X$ and $Y$ are consistent, $f_X(v)\neq f_Y(v)\iff set_X(v)\neq set_Y(v)$ for all vertices $v$,
where $set_X$ is the function on $G$ defined as $set_X(v)=1$ if $v\in X$ and $0$ otherwise. 

Fix a distribution $\mathcal{D}_{x_0,L}$ for snakes of length $L$ starting at $x_0$
and a distribution $\mathbf{D}_L$ on the set $\set{0,\ldots,L-1}$. 
With these in place, we let $\pr[j,X]{\cdot}$ denote the probability of an event over the distribution determined by
independently selecting $j$ according to  $\mathbf{D}_L$ and $X$ from $\mathcal{D}_{x_0,L}$.

We record Aaronson's definition of \emph{good} snakes, replacing the uniform distribution on the set $\set{0,\ldots,L-1}$ with the distribution $\mathbf{D}_L$,
\fix{and requiring a good snake's endpoint to be different from those of most other snakes}. 
\begin{definition}
A snake $X\in \mathcal{D}_{x_0,L}$ is \emph{$\epsilon$-good} w.r.t. distribution $\mathbf{D}_L$
if it satisfies the following:
\begin{enumerate}
\item $X$ is \emph{$0.9$-consistent}: 
$\pr[j,Y]{\mbox{$X$ and $Y$ are consistent, \fix{and $x_L\neq y_L$}}\mid Y_{0\to j}=X_{0\to j}} \geq 0.9\,.
$
\item $X$ is \emph{$\epsilon$-hitting}:
For all $v\in G$,
$
\pr[j,Y]{v\in Y_{j+1\to L} \mid Y_{0\to j}=X_{0\to j}} \leq \epsilon\,.
$
\end{enumerate}
\end{definition}

Our lower bounds will depend on the following adaptation of Aaronson's theorem of~\cite{Ref_Aaronson04lower}:

\begin{theorem}\label{Theorem_Aaronson_Snakes} 
Assume a snake $X$ drawn from $\mathcal{D}_{x_0,L}$ is $\epsilon$-good w.r.t. $\mathbf{D}_L$ with probability at least $0.9$.
Then 
\[
\mbox{$\RLS(G)=\Omega(1/\epsilon)$ and $\QLS(G)=\Omega(\sqrt{1/\epsilon})$.}
\]
\end{theorem}
\begin{proof}
To begin, we reduce the local search problem to a decision problem. 
For each snake $X\in \mathcal{D}_{x_0,L}$ and a bit $b \in \{0,1\}$,
define the function $g_{X,b}$ on $G$ as follows: $g_{X,b}(v)=(f_X(v),-1)$ for all vertices $v\neq x_L$, and $g_{X,b}(x_L)=(0,b)$.
Then, an input of the decision problem for local search on $G$ is an ordered \fix{pair $(X,g_{X,b})$,} 
where $X\in \mathcal{D}_{x_0,L}$ and $b \in \{0,1\}$ is an answer bit.
However, the ``snake part'' $X$ in the input cannot be  queried---it appears in the input as a bookkeeping tool. 
Given such an input $(X,g_{X,b})$, the decision problem is to output the answer bit $b$. 
Observe that the randomized (resp. quantum) query complexity of the decision problem is a lower bound for that of the original local search problem. 
This incorporation of $X$ into the input of the decision problem induces a natural one-to-one correspondence 
between an input set of the same answer bit and the set of snakes appearing in the input set. (Thanks to Scott Aaronson for suggesting this convention to us!)
In Aaronson's original version, since the input part $X$ is omitted, the snakes must be non-self-intersecting in order to obtain such a one-to-one correspondence. \newfix{Santha and Szegedy \cite{Ref_Santha04quantum} have presented an alternate approach for eliminating self-intersecting snakes while following Aaronson's proof scheme, though their technique 
only applies to the quantum case.} 

The remaining part of the proof, which establishes lower bounds for the decision problem \newfix{using the rational and quantum adversary methods}, is similar to Aaronson's proof with the exception of some technical details due to the adjustments in the definition of good snakes. We have relegated the full proof to the appendix.

\end{proof}

\section{Lower Bounds for Vertex-Transitive Graphs}
For simplicity, we first apply the snake framework for Cayley graphs,
and then extend the approach for vertex-transitive graphs.

\subsection{Lower bounds for Cayley graphs}
Consider a Cayley graph $C(G,\Gamma)$ of group $G$ determined by a generating set 
$\Gamma$. Our goal is to design a good snake distribution for $C(G,\Gamma)$. 
Our snakes will consist of a series of ``chunks'' so that the endpoint of each chunk
looks almost random given the preceding chunks. 
The locations at which a snake flicks its tail will be chosen
randomly from the locations of the chunks' endpoints. Each chunk is an ``extended'' shortest path
connecting its endpoint with the endpoint of the previous one. The relevant properties of these snakes depends on the length of each chunk as well as the number of chunks in each snake.
To determine these parameters, we begin with the following definitions.

Let $B(s)$ be the ball of radius $s$ centered at the group identity, i.e.,
$B(s)$ is the set of vertices $v$ for which $\Delta(1,v)\leq s$. 
We say that Cayley graph $C(G,\Gamma)$ is \emph{$s$-mixing} if there is a distribution
over the ball $B(s)$ that is $O\left({s}/{|G|^{3/2}}\right)$-close to the uniform distribution over $G$.
Clearly, every Cayley graph of diameter $d$ is $d$-mixing.

Now we assume $C(G,\Gamma)$ is \emph{$s$-mixing}, and let $D_s$ be a distribution
over $B(s)$ so that the extension of $D_s$ to be over $G$ is $\delta$-close to the uniform distribution over $G$,
where $s\leq \sqrt{|G|}$ and $\delta={0.1 s}/{|G|^{3/2}}$.
For each group element $g\in B(s)$, we fix a shortest path $(1,g_1,\ldots,g_r)$ in $C(G,\Gamma)$ from the group identity to $g$ (here $r=\Delta(1,g)$).
Then let $S(g)$ denote the sequence $(g_1,g_2,\ldots,g_s)$,  where $g_i=g$ for $i\geq r$.

Fix $\ell = {\sqrt{|G|}}/{(200s)}$ and let $L=(\ell+1)s$. 
We formally define our snake distribution $\mathcal{D}_{x_0,L}$ for snakes $X=(x_0,\ldots,x_L)$
as follows. For any $k\in\set{0,\ldots,\ell}$, 
choose $g_k$ independently according to the distribution $D_s$, and let the $k$th ``chunk''
$(x_{sk+1},\ldots,x_{sk+s})$ be identical to the sequence $x_{sk}S(g_k)$.

\begin{proposition}
A snake $X$ drawn from $\mathcal{D}_{x_0,L}$  $\delta$-mixes by $s$ steps in the sense that
for any $k$ and any $t\geq s$, $x_{sk+t}$ is $\delta$-close to uniform over $G$ given $x_{sk}$.
\end{proposition}

We define distribution $\mathbf{D}_L$ on $\set{0,\ldots,L-1}$  
as the uniform distribution on the set $\set{s, 2s,\ldots, \ell s}$.
So, unlike Aaronson's snakes whose tails \fix{may} be flicked at any location,
our snakes can not ``break'' in the middle of any chunk and only flick their tails at the chunk endpoints.

To show that most of our snakes are good, we start by showing that most snakes $X$ and $Y$ are consistent \fix{and have different endpoints}.

\begin{proposition}
Let $j$ be chosen according to $\mathbf{D}_L$. 
Let $X, Y$ be drawn from $\mathcal{D}_{x_0,L}$ 
conditioned on $Y_{0\to j} = X_{0\to j}$. 
Then
\[
\pr[X,j,Y]{\mbox{$X$ and $Y$ are consistent, and $x_L\neq y_L$}~|~Y_{0\to j} = X_{0\to j}} 
\geq 0.9999-\frac{2}{|G|}\,.
\]
\end{proposition}
\begin{proof}
Fix $j\in\set{s,2s,\ldots,\ell s}$. 
Suppose $v$ is a disagreement between $X$ and $Y$, letting $t=\max\set{i:x_i=v}$ and $t'=\max\set{i:y_i=v}$,
then $t\neq t'$ and $t',t\geq j$. We can't have both $t<j+s$ and $t'<j+s$, because
otherwise we would have $v\neq x_{j+s}$ and $v\neq y_{j+s}$ which implies
that both $t-j$ and $t'-j$ equal the distance from $x_{j}$ to $v$. 
 
If there is a disagreement, there must exist $t$ and $t'$ such that
$x_t=y_{t'}$ and either $t\geq j+s$ or $t'\geq j+s$.
In the case $t\geq j+s$, we have $x_t$ is $\delta$-close to uniform given $y_{t'}$, which implies
\[
\pr[X,Y_{j\to L}]{x_t = y_{t'}}
\leq \delta+\frac{1}{|G|} \leq\frac{2}{|G|}\,.
\]
Similarly, in the case $t'\geq j+s$, we also have
$\pr[X,Y_{j\to L}]{x_t = y_{t'}} \leq \frac{2}{|G|}$.
Summing up for all possible pairs of $t$ and $t'$ yields
\[
\pr[X,Y_{j\to L}]{\mbox{there is a disagreement between $X$ and $Y$}} \leq \frac{2(L-s)^2}{|G|}\leq 0.0001\,. 
\]
Averaging over $j$ produces
\[
\pr[X,j,Y]{\mbox{$X$ and $Y$ are not consistent}~|~Y_{0\to j}=X_{0\to j}} \leq 0.0001\,.
\]
\fix{
To complete the proof, observe that 
\[
\pr[X,j,Y]{x_L= y_L~|~Y_{0\to j}=X_{0\to j}}\leq \delta+\frac{1}{|G|}\leq \frac{2}{|G|}\,.
\]
since $y_L$ is $\delta$-close to uniform given $x_L$.}
\end{proof}

By Markov's inequality, we obtain:
\begin{corollary}Let $X$ be drawn from $\mathcal{D}_{x_0,L}$. Then
\[
\pr[X]{\mbox{$X$ is $0.9$-consistent}} \geq 1-\frac{0.0001+2/|G|}{0.1} = 0.999-\frac{20}{|G|}\,.
\]
\end{corollary}

We now turn our attention to bounding the hitting probability when a snake flicks its tail.
Following Aaronson, we introduce a notion of \emph{$\epsilon$-sparseness} for snakes and show that
\emph{(i)} if a snake is $\epsilon$-sparse  then it is $O(\epsilon)$-hitting, and that
\emph{(ii)} most snakes are $\epsilon$-sparse. 

Formally, we define:

\begin{definition}
For each $x\in G$, let $P(x) = \pr[g\in D_s]{x\in S(g)}$.
A snake $X$ drawn from $\mathcal{D}_{x_0,L}$ is called \emph{$\epsilon$-sparse} if
for all vertex $v\in G$,
\[
\sum_{k=1}^{\ell} P(x_{sk}^{-1}v) \leq \epsilon\ell\,.
\] 
\end{definition}
 
Intuitively, the sparseness of a snake means that if the snake flicks a random chunk, 
it is unlikely to hit any fixed vertex.

\begin{proposition}
For $\epsilon\geq \frac{2(L-s)}{|G|}$, if snake $X$ is $\epsilon$-sparse then $X$ is $2\epsilon$-hitting.
\end{proposition}
\begin{proof}
Fix a snake $X$, and fix $j\in \set{s,2s\ldots, \ell s}$. 
Let $Y$ be drawn from $\mathcal{D}_{x_0,L}$ conditioned on the event that 
$Y_{0\to j}=X_{0\to j}$. 
Since $y_{t}$ is $\delta$-close to uniform for all $t\geq j+s$, 
\[
\pr[Y]{v\in Y_{j+s\to L}~|~ Y_{0\to j}=X_{0\to j}} \leq (L-s)(\delta+\frac{1}{|G|}) \leq \frac{2(L-s)}{|G|}\,.
\]
On the other hand, 
\[
\pr[Y]{v\in Y_{j+1\to j+s}~|~ Y_{0\to j}=X_{0\to j}}=\pr[g\in D_s]{v\in x_{j}S(g)}= P(x_{j}^{-1}v)\,.
\]
Hence,
\[
\pr[j,Y]{v\in Y_{j+1\to L}~|~ Y_{0\to j}=X_{0\to j}} \leq \frac{1}{\ell}\sum_{k=1}^{\ell}P(x_{sk}^{-1}v)+\frac{2(L-s)}{|G|} \leq 2\epsilon\,.
\]
\end{proof}

It remains to show that a snake drawn from $\mathcal{D}_{x_0,L}$
is $\epsilon$-sparse with high probability. 
Firstly, we consider for the ``ideal'' case in which the endpoints of the chunks in a snake are independently uniform.

\begin{lemma}\label{Lemma_Sparse_Uniform}
Let $u_1,\ldots, u_{\ell}$ be independently and uniformly random vertices in $G$.
If $\frac{s}{|G|} \leq \epsilon^2/6$ then
\[
\pr[u_1,\ldots, u_{\ell}]{\sum_{i=1}^{\ell}P(u_i)>2\ell\epsilon} \leq 2^{-\ell\epsilon}\,.
\]
\end{lemma}
\begin{proof}
We will use a Chernoff bound to show that there are very few $u_i$'s for which $P(u_i)$ is large.
To do this, we first need an upper bound on the expectation of $P(u_i)$.
Let $u$ be a uniformly random vertex in $G$. 
For any given $g\in G$, we have $\pr[u]{u\in S(g)}=\frac{\Delta(1,g)}{|G|}\leq \frac{s}{|G|}$ .
Averaging over $g\in D_s$ yields \(\pr[g,u]{u\in S(g)}\leq \frac{s}{|G|}\),
where $g$ is chosen from $D_s$ independently to $u$.
Since $\expect[u]{P(u)}=\pr[u,g]{u\in S(g)}$, we have $\expect[u]{P(u)}\leq \frac{s}{|G|}$.

Let $Z = |\set{i: P(u_i) \geq \epsilon}|$. 
By Markov's inequality,
\[
\expect{Z} = \ell\pr[u]{P(u)\geq \epsilon} \leq \frac{\ell \expect[u]{P(u)}}{\epsilon} \leq \frac{\ell s}{|G|\epsilon} = \mu\,.
\]
By a Chernoff bound, for any $\lambda \geq 2e$
\[
\pr[u]{Z \geq \lambda\mu } \leq \left(\frac{e^{\lambda-1}}{\lambda^{\lambda}}\right)^{\mu}
  = \left(\frac{e}{\lambda}\right)^{\lambda\mu} e^{-\mu} \leq 2^{-\lambda\mu-\mu}\,.
\]
Note that if $Z<\lambda\mu$ then 
\[
\sum_{i=1}^{\ell} P(u_i) \leq (\ell-Z)\epsilon + Z \leq \ell\epsilon+\lambda\mu\,.
\]
Setting $\lambda\mu = \ell\epsilon$, which satisfies $\lambda \geq 2e$ due to the assumption that $\frac{s}{|G|} \leq \epsilon^2/6$, 
we have
\[
\pr[u_1,\ldots,u_\ell]{\sum_{i=1}^{\ell} P(u_i)>2\ell\epsilon}\leq \pr[u_1,\ldots,u_\ell]{Z \geq \ell\epsilon } \leq 2^{-\ell\epsilon}\,.
\]
\end{proof}

In order to apply this to our scenario without strict independence, we record the following fact about distance in total variation.

\begin{proposition}\label{Prop_TotalVariance}
Let $X_1,\ldots,X_n$ and $Y_1,\ldots,Y_n$ be discrete random variables so that
$X_i$ and $Y_i$ have the same value range.
Let $(X_i \mid A_1,\ldots,A_{i-1})$ denote the distribution of $X_i$ given  that $X_1\in A_1,\ldots, X_{i-1}\in A_{i-1}$; similarly
let $(Y_i \mid A_1,\ldots,A_{i-1})$ denote the distribution of $Y_i$ given that $Y_1 \in A_1,\ldots, Y_{i-1}\in A_{i-1}$.
Then
\[
\|(X_1,\ldots,X_n)-(Y_1,\ldots,Y_n)\|_{t.v.} \leq \|X_1-Y_1\|_{t.v.} +  \sum_{i=2}^n \Delta_i
\] 
where
\[
\Delta_i = \max_{A_1,\ldots,A_{i-1}}\|(X_i\mid A_1,\ldots,A_{i-1})-(Y_i\mid A_1,\ldots,A_{i-1})\|_{t.v.}\,.
\]
\end{proposition}

A detailed proof of Proposition \ref{Prop_TotalVariance} can be found in the appendix.

\begin{lemma}
Suppose $\frac{s}{|G|} \leq \epsilon^2/6$.
Then a snake $X$ drawn from $\mathcal{D}_{x_0,L}$ is $2\epsilon$-sparse with probability at least $1-|G|2^{-\ell\epsilon}-1/2000$.
\end{lemma}
\begin{proof}
The proof for the lemma follows immediately by observing that for any vertex $v$,
the variables $x_{s}^{-1}v,\ldots,x_{s\ell}^{-1}v$ satisfy that 
$x_{s(k+1)}^{-1}v$ is $\delta$-close to uniform given $x_{sk}^{-1}v$.
By Proposition~\ref{Prop_TotalVariance}, 
\[
\left|\pr[X]{\sum_{k=1}^{\ell}P(x_{sk}^{-1}v)>2\ell\epsilon} - \pr[u_1,\ldots, u_{\ell}]{\sum_{i=1}^{\ell}P(u_i)>2\ell\epsilon}\right| 
\leq \ell\delta \leq \frac{1}{2000|G|}\,.
\]
From Lemma \ref{Lemma_Sparse_Uniform},
\[
\pr[X]{\sum_{k=1}^{\ell}P(x_{sk}^{-1}v)>2\ell\epsilon} \leq 2^{-\ell\epsilon}+ \frac{1}{2000|G|}\,.
\]
Summing up over $v\in G$ gives $\pr[X]{\mbox{$X$ is \emph{not} $2\epsilon$-sparse}}  \leq |G|2^{-\ell\epsilon}+ 1/2000$.
\end{proof}

We need to choose $\epsilon$ such that $|G|2^{-\ell\epsilon}\leq 1/2000$, or  
$\epsilon\geq \frac{\log|G|+O(1)}{\ell}$.

\begin{corollary}
A snake $X$ drawn from $\mathcal{D}_{x_0,L}$ is
$O\left(\frac{s\log|G|}{\sqrt{|G|}}\right)$-hitting with probability at least $0.999$.
\end{corollary}

Putting all the pieces together and applying Theorem \ref{Theorem_Aaronson_Snakes}, we have

\begin{theorem}
For $s=O(\sqrt{|G|})$, if Cayley graph $C(G,\Gamma)$ is $s$-mixing, then
\[
\RLS(C(G,\Gamma)) = \Omega\left(\frac{\sqrt{|G|}}{s\log|G|}\right),
~~~~~
\QLS(C(G,\Gamma)) = \Omega\left(\frac{\sqrt[4]{|G|}}{\sqrt{s\log|G|}}\right)\,.
\]
\end{theorem}

In particular, any Cayley graph $C(G,\Gamma)$ of diameter $d$ has
\[
\RLS(C(G,\Gamma)) = \Omega\left(\frac{\sqrt{|G|}}{d\log|G|}\right),
~~~~~~
\QLS(C(G,\Gamma)) = \Omega\left(\frac{\sqrt[4]{|G|}}{\sqrt{d\log|G|}}\right)\,.
\]
For comparison, applying Aldous's randomized upper bound \cite{Ref_Aldous83minimization} and Aaronson's quantum upper bound \cite{Ref_Aaronson04lower} for arbitrary Cayley graph $C(G,\Gamma)$, we have 
\[
\mbox{$\RLS(C(G,\Gamma))=O\left(\sqrt{|G||\Gamma|}\right)$ and $\QLS(C(G,\Gamma))=O\left(\sqrt[3]{|G|}\sqrt[6]{|\Gamma|}\right)$.}
\]
For example, for constant degree expanding Cayley graphs, this randomized lower bound is tight to within $O(\log^2|G|)$ of Aldous's upper bound.

\paragraph{Random Cayley graphs.} In fact, it can be showed that most Cayley graphs are $s$-mixing for $s=\Omega(\log|G|)$.
Let $g_1,\ldots, g_s$ be a sequence of group elements.
Following \cite{Ref_Babai91local}, we call an element of the form
$g_1^{a_1}\cdots g_s^{a_s}$, where $a_i\in\set{0,1}$, a \emph{subproduct} of
the sequence $g_1,\ldots, g_s$. A \emph{random subproduct} of this sequence is a subproduct 
obtained by independently choosing $a_i$ as a fair coin flip. 
A sequence $g_1,\ldots, g_s$ is called \emph{a sequence of $\delta$-uniform Erd\"os-R\'enyi (E-R) generators}
if its random subproductors are $\delta$-uniformly distributed over $G$ in the sense that
\[
(1-\delta)\frac{1}{|G|} \leq \pr[a_1,\ldots,a_s]{g_1^{a_1}\cdots g_s^{a_s} = g}\leq  (1+\delta)\frac{1}{|G|} \qquad \text{for all $g\in G$}\,.
\]

\begin{theorem}(Erd\"os and R\'enyi, See also \cite{Ref_Babai91local})
For $s\geq 2\log|G|+2\log(1/\delta)+\lambda$, a sequence of $s$ random elements of $G$
is a sequence of $\delta$-uniform E-R generators with probability at least $1-2^{-\lambda}$.
\end{theorem}

Clearly, any Cayley graph determined by an $s$-length sequence of $\delta$-uniform E-R generators is $s$-mixing.
So applying our lower bounds for arbitrary Cayley graphs and the E-R theorem, we have
\begin{proposition}
Let $s\geq 5\log|G|-2\log s+\lambda$. With probability at least $1-2^{-\lambda}$, a random Cayley graph $C(G,\Gamma)$ determined by a sequence
of $s$ random group elements has
\[
O(\sqrt{|G|s}) \geq \RLS(C(G,\Gamma)) \geq \Omega\left(\frac{\sqrt{|G|}}{s\log|G|}\right)
\;\text{and}\;
O(\sqrt[3]{|G|}\sqrt[6]{s}) \geq \QLS(C(G,\Gamma)) \geq \Omega\left(\frac{\sqrt[4]{|G|}}{\sqrt{s\log|G|}}\right)\,.
\]
\end{proposition}

\subsection{Extending to Vertex-Transitive Graphs}
Our approach above for Cayley graphs can be easily extended to vertex-transitive graphs.
We shall describe here how to define a snake distribution $\mathcal{D}_{x_0,L}$ similar to that for a Cayley graph.
Consider a vertex-transitive graph $G=(V,E)$ with $N=|V|$, and let $d$ be the diameter of $G$.
We fix an arbitrary vertex $v_0\in V$. For each vertex $v\in V$, we also fix an
extended shortest path $S(v)=(v_1,\ldots,v_d)$ of length $d$ from $v_0$ to $v$. ($v_0$ is omitted in $S(v)$ for technical reasons.) 
That is, $(v_0,\ldots,v_r)$ is the actual shortest path from $v_0$ to $v$,
where $r=\Delta(v_0,v)$, and $v_i=v$ for all $i\geq r$.

Since the automorphism group of $G$ acts transitively on $V$,
we can fix an automorphism $\sigma_x$, for each $x\in V$, so that $\sigma_x(v_0)=x$.
Hence, for any $x,v\in V$, the sequence $\sigma_xS(v)$ is
the extended shortest path from $x$ to $\sigma_x(v)$. 
So now we can determine the $k$th chunk of a snake as the sequence $\sigma_{x_{dk}}S(u_k)$,
where $x_{dk}$ is the endpoint of the $(k-1)$th chunk of the snake, and $u_k$ is an independently and uniformly random vertex.
Let $P(x)=\pr[u]{x\in S(u)}$, where $u$ is chosen from $V$ uniformly at random.
The condition for a snake $X=(x_0,\ldots,x_{(\ell+1)d})$ 
to be \emph{$\epsilon$-sparse} is now redefined as
\[
\sum_{k=1}^{\ell} P\left(\sigma^{-1}_{x_{dk}}(v)\right)\leq \ell\epsilon ~~~\mbox{for all $v\in V$.}
\]

Observe that, given $x_{dk}$, the endpoint $x_{dk+k}=\sigma_{x_{dk}}(u_k)$ of the $k$th chunk 
is a uniformly random vertex, since $\sigma_{x_{dk}}$ is a bijective.
Also clearly, $\sigma_x \neq \sigma_y$ for any $x\neq y$ because
$\sigma_x$ and $\sigma_y$ send $v_0$ to different places.
This means there is a one-to-one correspondence $x\leftrightarrow \sigma_x$ 
between $V$ and the set of automorphisms $\set{\sigma_x:x\in V}$.
Therefore, if $x$ is uniformly distributed over $V$, then so is the vertex at any given position
in $\sigma_x S$, for any sequence $S$ of vertices. 
It follows that in our snake $X=(x_0,\ldots,x_L)$, for all $t\geq k$, $x_{dk+t}$ is uniformly distributed over $V$ given $x_{dk}$.
With this snake distribution, we can similarly follow the proof for Cayley graphs to prove the lower bounds
for vertex-transitive graphs as given in Theorem \ref{Theorem_Vertex-Transitive}.

\subsection*{Acknowledgements}

We gratefully acknowledge Scott Aaronson for discussing his previous work with us and showing us the trick for removing the requirement of snake non-self-intersection. \newfix{We would like to thank anonymous referees for many helpful comments.}

\bibliography{Localsearch}
\bibliographystyle{plain}

\newpage
\appendix
\section{Appendix}

\subsection{Quantum and Relational Adversary Methods}

The quantum adversary method~\cite{Ref_Ambainis00quantum} is a powerful tool underlying
many proofs of quantum lower bounds. The classical counterpart applied above is the relational adversary method \cite{Ref_Aaronson04lower}. The central intuition of these adversary methods is to make it hard to distinguish
``related'' input sets. Technically, consider two input sets $\mathcal{A}$ and $\mathcal{B}$ for a function 
{$F:I^n\to [m]$} so that $F(A)\neq F(B)$ for all $A\in\mathcal{A}$ and $B\in\mathcal{B}$. Here, an input to function $F$ is a black-box function $A:[n]\to I$. The oracle for an input $A$ answers queries of the form $A(x)=?$. \fix{If $A$ and $B$ are the two inputs that have the same value at every queryable location, then we must have $F(A)=F(B)$.}
Define a ``relation'' function $R(A,B)\geq 0$ on $\mathcal{A}\times \mathcal{B}$.
Two inputs $A$ and $B$ are said to be related if $R(A,B)>0$.
Then, for $A\in\mathcal{A}$, $B\in\mathcal{B}$, and a \fix{queryable} location $x\in [n]$, let
\[
\begin{array}{ll}
M(A)   = \sum_{B'\in\mathcal{B}} R(A,B'),                  & M(B)   = \sum_{A'\in\mathcal{A}} R(A',B)\\
M(A,x) = {\sum_{B'\in\mathcal{B}:A(x)\neq B'(x)} R(A,B')}, & M(B,x) = {\sum_{A'\in\mathcal{A}:A'(x)\neq B(x)} R(A',B)}\,.\\
\end{array}
\]  
Intuitively, the fraction \fix{$M(A,x)/M(A)$ (resp. $M(B,x)/M(B)$) measures how hard it is to distinguish input $A$  (resp. $B$) with related inputs in $\mathcal{B}$ (resp. $\mathcal{A}$) by queying at location $x$.} 
Formally, if there are such input sets $\mathcal{A}, \mathcal{B}$ and relation function $R(A,B)$, then

\begin{theorem}(Ambainis) The number of quantum queries needed to evaluate $F$ with probability at least $0.9$
is $\Omega(M_{\rm geom})$, where
\[
M_{\rm geom} = \mathop{\min_{A\in\mathcal{A}, B\in\mathcal{B}, x}}_{R(A,B)>0, A(x)\neq B(x)} \sqrt{\frac{M(A)}{M(A,x)}\frac{M(B)}{M(B,x)}}\,.
\]
\end{theorem}

\begin{theorem}(Aaronson) The number of randomized queries needed to evaluate $F$ with probability at least $0.9$
is $\Omega(M_{\rm max})$, where
\[
M_{\rm max} = \mathop{\min_{A\in\mathcal{A}, B\in\mathcal{B}, x}}_{R(A,B)>0, A(x)\neq B(x)} \max\set{\frac{M(A)}{M(A,x)},~ \frac{M(B)}{M(B,x)}}\,.
\]
\end{theorem}

\remove{
Zhang \cite{Ref_Zhang05power} recently proposed an improvement of Ambainis' quantum adversary method by 
replacing $M(A,x)$ and $M(B,x)$ in the formula of $M_{\rm geom}$ with
\[
M^*(A,x) = {\sum_{B'\in\mathcal{B}:A(x)\neq B'(x)} \alpha(A,B',x)},~~~~
M^*(B,x) = {\sum_{A'\in\mathcal{A}:A'(x)\neq B(x)} \beta(A',B,x)}\,.\\
\]
where $\alpha(A,B,x)$ and $\beta(A,B,x)$ are positive functions such that
$\alpha(A,B,x)\beta(A,B,x)\geq R^2(A,B)$ for all $A\in\mathcal{A}$, $B\in\mathcal{B}$, and $x\in[n]$.
}

\subsection{Proofs}
\subsubsection{Continued proof for Aaronson's theorem (Theorem~\ref{Theorem_Aaronson_Snakes})}
\begin{proof}
To apply the quantum and relational adversary method for the decision problem, define the input sets 
\fix{$\mathcal{A}=\set{(X,g_{X,0}): X\in \mathcal{D}^*}$ and 
$\mathcal{B}=\set{(Y,g_{Y,1}): Y\in \mathcal{D}^*}$,} where $\mathcal{D}^*$ denotes the set of $\epsilon$-good snakes drawn from $\mathcal{D}_{x_0,L}$. 
For simplicity, we write $A_X$ as $(X,g_{X,0})$, and $B_Y$ as $(Y,g_{Y,1})$.
For $A_X\in\mathcal{A}$ and $B_Y\in \mathcal{B}$, define relation function $R(A_X,B_Y) = w(X,Y)$ if $X$ and $Y$ are consistent \fix{and $x_L\neq y_L$}, 
and $R(A_X,B_Y) = 0$ otherwise, where $w(X,Y)$ is determined as follows.
Let  $p(X)$ be the probability of drawing snake $X$ from $\mathcal{D}_{x_0,L}$, 
and  let 
\[
w(X,Y) = p(X) \pr[j,Z]{Z = Y ~|~ Z_{0\to j}=X_{0\to j}}\,.
\]
\begin{claim} For any snakes $X,Y\in\mathcal{D}_{x_0,L}$,
we have $w(X,Y) = w(Y,X)$.
\end{claim}
\begin{proof}(of the claim) 
Fix $j\in \set{0,\ldots,L-1}$ and let $q_j(X,Y) = \pr[Z]{Z = Y ~|~ Z_{0\to j}=X_{0\to j}}$.
We want to show
\[
p(X)q_j(X,Y)=p(Y)q_j(Y,X)\,.
\]
Assume $X_{0\to j}=Y_{0\to j}$, otherwise $q_j(X,Y)=q_j(Y,X)=0$.
Then letting $Z$ be drawn from $\mathcal{D}_{x_0,L}$ and let $E$ be the event $Z_{0\to j}=X_{0\to j}=Y_{0\to j}$,
we have
\[
\begin{split}
p(X)q_j(X,Y) &= \pr[Z]{E}\cdot\pr[Z]{Z_{j+1\to L}=X_{j+1\to L} | E}\cdot \pr[Z]{Z_{j+1\to L}=Y_{j+1\to L} | E}\\ 
 &= \pr[Z]{E} \cdot\pr[Z]{Z_{j+1\to L}=Y_{j+1\to L} | E}\cdot\pr[Z]{Z_{j+1\to L}=X_{j+1\to L} | E}\\ 
&= p(Y)q_j(Y,X) \,.
\end{split}
\]
completing the proof for the claim.
\end{proof}
As in Aaronson's original proof, we won't be able to take the whole input sets
$\mathcal{A}$ and $\mathcal{B}$ defined above because of the fact that not all snakes are good.
Instead, we will take only a subset of each of these input sets that would be hard enough to distinguish.
This is done by applying Lemma 8 in \cite{Ref_Aaronson04lower}, which states as follows.
\begin{lemma}\label{Aaronson_Lemma_8}
Let $p(1),\ldots,p(m)$ be positive reals such that $\sum_{i}p(i)\leq 1$.
Let $R(i,j)$, for $i,j\in\set{1,\ldots,m}$, be nonnegative reals satisfying $R(i,j)=R(j,i)$ and
$\sum_{i,j}R(i,j)\geq r$. Then there exists a nonempty subset $U\in \set{1,\ldots,m}$ such that
$\sum_{j\in U} R(i,j)\geq rp(i)/2$ for all $i\in U$.
\end{lemma}
To apply this lemma, we need  a lower bound for the sum $\sum_{X,Y\in \mathcal{D^*}}R(A_X, B_Y)$.
Let $E(X,Y)$ denote the event that snakes $X$ and $Y$ are consistent \fix{and $x_L\neq y_L$}. 
For any $X\in\mathcal{D}^*$, we have
\[
\sum_{Y: E(X,Y)} w(X,Y) = p(X)\pr[j,Y]{E(X,Y)~|~Y_{0\to j}=X_{0\to j}} \geq 0.9p(X)\,.
\]
Hence, since a snake drawn from $\mathcal{D}_{x_0,L}$ is good with probability at least $0.9$,
\[
\sum_{X,Y: E(X,Y)} w(X,Y) \geq 0.9\sum_{X\in \mathcal{D}^*}p(X) \geq 0.9\times 0.9 \geq 0.8\,.
\]
By the union bound,
\[
\sum_{X,Y\in \mathcal{D^*}}R(A_X, B_Y) \geq \sum_{X,Y: E(X,Y)} w(X,Y) - \sum_{X\not\in \mathcal{D}^*} p(X) - \sum_{Y\not\in \mathcal{D}^*} p(Y)
\geq 0.8-0.1 -0.1 = 0.6\,.
\]
So, by Lemma \ref{Aaronson_Lemma_8}, there exists a nonempty subset 
$\widetilde{\mathcal{D}}\subset \mathcal{D}^*$ so that for all $X, Y\in \widetilde{\mathcal{D}}$,
\begin{align*}
\sum_{Y'\in \widetilde{\mathcal{D}}} R(A_X, B_{Y'}) &\geq 0.3p(X)\,,\\
\sum_{X'\in \widetilde{\mathcal{D}}} R(A_{X'}, B_{Y}) &\geq 0.3p(Y)\,.
\end{align*}
So now we take the input sets $\widetilde{\mathcal{A}}=\set{A_X: X\in\widetilde{D}}$
and $\widetilde{\mathcal{B}}=\set{B_Y: Y\in\widetilde{D}}$. 
We have shown that $M(A_X)\geq 0.3p(X)$ and $M(B_Y)\geq 0.3p(Y)$ for any $A_X\in\widetilde{\mathcal{A}}$ and  $B_Y\in \widetilde{\mathcal{B}}$.
Since the snake part in the inputs can not be queried, we only 
care about the measure for distinguishing $A_X, B_Y$ with their related inputs by querying the function part (i.e. $g_{X,0}$ or $g_{Y,1}$) in the inputs.
Formally, we focus on lower-bounding $M(A_X,v)$ and $M(B_Y,v)$ for inputs $A_X\in\widetilde{\mathcal{A}}, B_Y\in \widetilde{\mathcal{B}}$
for which $R(A_X,B_Y)>0$ and $g_{X,0}(v)\neq g_{Y,1}(v)$. 
\fix{We remark that since $R(A_X,B_Y)>0$, the event $E(X,Y)$ must hold, which implies that for all vertex $v$, 
$$g_{X,0}(v)\neq g_{Y,1}(v)\iff f_X(v)\neq f_Y(v) \iff set_X(v)\neq set_Y(v)\,.$$} 
Applying the quantum and randomized adversary method, we will have 
$\RLS(G)\geq \Omega(M_{\rm max})$ and $\QLS(G)\geq \Omega(M_{\rm geom})$, where
\[
M_{\rm max} = \mathop{\min_{A_X\in\widetilde{\mathcal{A}}, B_Y\in\widetilde{\mathcal{B}}, v}}_{R(A_X,B_Y)>0, set_X(v)\neq set_Y(v)} \max\set{\frac{M(A_X)}{M(A_X,v)},~\frac{M(B_Y)}{M(B_Y,v)}}
\]
\[
M_{\rm geom} = \mathop{\min_{A_X\in\widetilde{\mathcal{A}}, B_Y\in\widetilde{\mathcal{B}}, v}}_{R(A_X,B_Y)>0, set_X(v)\neq set_Y(v)} \sqrt{\frac{M(A_X)}{M(A_X,v)}\frac{M(B_Y)}{M(B_Y,v)}}\,.
\]

Let $A_X\in\widetilde{\mathcal{A}}, B_Y\in \widetilde{\mathcal{B}}$ be inputs
for which \fix{$set_X(v)\neq set_Y(v)$}. \fix{Then $v\not\in X$ or $v\not\in Y$.} 
Assuming the case $v\not\in X$, we will show $M(A_X,v)$ is small.
We have
\[
\begin{split}
M(A_X,v) 
&\leq \sum_{Y'\in \widetilde{\mathcal{D}}: set_X(v)\neq set_{Y'}(v)} w(X,Y') \\
&\leq \sum_{Y': v\in Y' } p(X)\pr[j,Z]{Z = Y' ~|~ Z_{0\to j}=X_{0\to j}} \\
&= p(X)\pr[j,Z]{v\in Z  ~|~ Z_{0\to j}=X_{0\to j}} \\
&= p(X)\pr[j,Z]{v\in Z_{j+1\to L} ~|~ Z_{0\to j}=X_{0\to j}} ~~~(\mbox{since $v\not\in X$})\\
&\leq p(X)\epsilon ~~~(\mbox{since $X$ is $\epsilon$-hitting})\,.
\end{split}
\]
In the case $v\not\in Y$, we can also obtain $M(B_Y,v)\leq p(Y)\epsilon$ due to symmetry.
Hence, 
\[
\max\set{\frac{M(A_X)}{M(A_X,v)},~\frac{M(B_Y)}{M(B_Y,v)}} \geq 0.3/\epsilon
\]
\[
\sqrt{\frac{M(A_X)}{M(A_X,v)}\frac{M(B_Y)}{M(B_Y,v)}} \geq \sqrt{0.3/\epsilon}\,.
\]
The latter inequality is obtained due to the fact that $M(A_X,v)\leq M(A_X)$ and  $M(B_Y,v)\leq M(B_Y)$.
Consequently, $M_{\rm max}=\Omega(1/\epsilon)$ and $M_{\rm geom}=\Omega(\sqrt{1/\epsilon})$, completing the proof for Theorem \ref{Theorem_Aaronson_Snakes}.
\end{proof}

\subsubsection{Proof of Proposition~\ref{Prop_TotalVariance}}
\begin{proof}
We prove by induction on $n$. The case $n=2$
can be easily obtained by applying the following simple fact:
\begin{fact}
Let $x_1, x_2, y_1, y_2$ be any real numbers in $[0,1]$. Then
\[
|x_1x_2 - y_1y_2| = |(x_1-y_1)x_2 + (x_2-y_2)y_1| \leq |x_1-y_1|x_2 + |x_2-y_2|y_1 \leq |x_1-y_1| + |x_2-y_2|\,.
\]
\end{fact}
In particular, applying the above fact, we have for any pair of events $(A,B)$,
\[
\begin{split}
\Bigl|\pr{X_1\in A, X_2\in B} - \pr{Y_1\in A, Y_2\in B}\Bigr| \leq 
&\Bigl|\pr{X_1\in A}-\pr{Y_1\in A}\Bigr| +  \\
&\Bigl|\pr{X_2\in B|X_1\in A}-\pr{Y_2\in B|Y_1\in A}\Bigr|\,.
\end{split}
\] 
Recall that by definition of total variation,
\[
\|(X_1,X_2) - (Y_1,Y_2)\|_{t.v.} = \max_{A,B}\Bigl|\pr{X_1\in A, X_2\in B}-\pr{Y_1\in A, Y_2\in B}\Bigr| \quad\text{and}
\]
\[
\Delta_2 = \max_{A,B}\Bigl|\pr{X_2\in B|X_1\in A}-\pr{Y_2\in B|Y_1\in A}\Bigr| \,.
\]
Hence,
\[
\|(X_1,X_2) - (Y_1,Y_2)\|_{t.v.} \leq \|X_1-Y_1\|_{t.v.} + \Delta_2\,.
\]
Now we can apply this result and get
\[
\|(X_1,\ldots,X_n)-(Y_1,\ldots, Y_n)\|_{t.v.} \leq \|(X_1,\ldots,X_{n-1})-(Y_1,\ldots, Y_{n-1})\|_{t.v.} + \Delta_n
\]
which establishes the proposition by induction.
\end{proof}

\subsection{Upper Bounds for Local Search}
Various upper bounds for both quantum and classical query complexities have been given for general graphs.
For any graph $G$ of $N$ vertices and maximal degree $\delta$, it has been showed that
$\RLS(G)=O(\sqrt{N\delta})$ \cite{Ref_Aldous83minimization} and $\QLS(G)=O(N^{1/3}\delta^{1/6})$ \cite{Ref_Aaronson04lower}. 
The idea for designing local search algorithms in \cite{Ref_Aldous83minimization,Ref_Aaronson04lower} is 
random sampling followed by steepest descent. More specifically, these algorithms start off by sampling a subset of vertices,
find the best vertex $v$ (i.e., the one with the minimum $f$ value) 
in the sampled set, and finally performing steepest descent beginning at the chosen vertex $v$.

Zhang \cite{Ref_Zhang06new} later introduced new quantum and randomized algorithms for local search
on general graphs, providing upper bounds that depend on the graph diameter and the expansion speed.
\remove{
The main ingredients in Zhang's algorithms also include sampling and steepest descent.
However, instead of running steepest descent right away after sampling, 
they repeat the sampling procedure, 
narrowing down the sample space after each iteration to a smaller region whose boundary is ``good''
in the sense that every vertex in the boundary has large $f$ values. 
Such a boundary is to guarantee the existence of a local minimum in the sample space.
This recursion is continued until the diameter of the sample space is within a constant.
Then, beginning at the vertex chosen from the last sampling iteration, they start the steepest descent procedure, 
which is guaranteed to terminate within a constant number steps due to 
the good boundary of the last sample space.

In particular, letting $c(k)$ be the maximal size of the balls of radius $k$ in the graph $G$, i.e., 
$c(k)=\max_{v}|\set{u:\Delta(v,u)\leq k}|$, Zhang showed that
\begin{theorem}(Zhang \cite{Ref_Zhang06new})
If there exists $\alpha\geq 1$ for which $c(k)=O(k^{\alpha})$ for all $k=1,\ldots,d$, then
\[
\RLS(G)=
\left\{\begin{array}{ll}
O(d^{\alpha-1}\log\log d) & {\rm if} ~ \alpha>1\\
O(\log d\log\log d) & {\rm if} ~ \alpha=1\,.\\
\end{array}\right.
~~~~
\QLS(G)=
\left\{\begin{array}{ll}
O(d^{(\alpha-1)/2}(\log\log d)^{3/2}) & {\rm if} ~ \alpha>1\\
O(\log d\log\log d) & {\rm if} ~ \alpha=1\,.\\
\end{array}\right.
\]
where $d$ is the diameter of the graph $G$.
\end{theorem}
}
While Zhang's upper bounds can only work well for graphs with slow expansion speed, such as hypecubes,
many vertex-transitive graphs, unfortunately, do not possess this property. 
Also, 
Zhang's randomized upper bound 
is no better than $O\left(\frac{N}{d}\log\log d \right)$, and his quantum upper bound is no better than 
$O\left(\sqrt{\frac{N}{d}}(\log\log d)^{1.5} \right)$, except for the line or cycle graphs,
where $d$ is the diameter of the graph.
This means Zhang's upper bounds do not seem to beat Aaronson and Aldous's bounds,
especially for graphs with small degrees and small diameters.
Note that there are Cayley graphs of non-abelian simple groups which have constant degrees 
and have diameters no larger than $O(\log N)$. 
While Zhang's upper bounds fail for graphs of small diameters, Aldous and Aaronson's upper bounds fail for graphs of large degrees.
So, a question to ask is whether there is a better upper bound for graphs with large degrees and small diameters?

Recently, Verhoeven \cite{Ref_Verhoeven06enhanced} has proposed another deterministic algorithm and enhanced Zhang's quantum algorithm,
improving upper bounds on deterministic and quantum query complexities of Local Search
that depend on the graph's degrees and genus. Precisely, he showed that for any $N$-vertex graph $G$
of genus $g$ and maximal degree $\delta$, 
the deterministic (thus, randomized) and query complexities of Local Search on $G$ are $\delta+O(\sqrt{g})\sqrt{N}$ and
$O(\sqrt{\delta})+O(\sqrt[4]{g})\sqrt[4]{N}\log\log N$, respectively.
However, these bounds fail for the class of graphs we are caring about: vertex-transitive graphs,
since every vertex-transitive graph is regular and it has been shown that
the genus of an $N$-vertex $m$-egde connected graph is at least $\ceil{\frac{m}{6}-\frac{N}{2}+1}$ (see \cite[p114]{Ref_Mohar01graph}).

\end{document}